\documentclass{amsart}
     
\newtheorem{theorem}{Theorem}[section]
\newtheorem{lemma}[theorem]{Lemma}

\theoremstyle{definition}

\theoremstyle{remark}
\newtheorem{remark}[theorem]{Remark}
 \usepackage{epsfig} 
 \usepackage{graphics} 
\usepackage{mathtools}
\numberwithin{equation}{section}

\def\g{{\gamma}}

\def\e{{\epsilon}}

\def\z{{\zeta}}

\def\n{{\nu}}
\def\f{{\phi}}

\def\m{{\mu}}

\def\p{{\psi}}

\def\ttt{{\theta}}

\def\G{{\Gamma}}
\def\o{{\omega}}

\def\aa{{{\mathcal A}}}
\def\bb{{{\mathcal B}}}

\def\ff{{{\mathcal F}}}

\def\uu{{{\mathcal U}}}
\def\vv{{{\mathcal V}}}

\def\xx{{{\mathcal X}}}
\def\yy{{{\mathcal Y}}}

\def\nn{{{\mathcal N}}}

\def\R{{{\bf R}^1}}

\def\CCC{{{\bf C}}}
\def\BBB{{{\bf B}}}
\def\XXX{{{\bf X}}}
\def\YYY{{{\bf Y}}}
\def\AAA{{{\bf A}}}
\def\000{{{\bf 0}}}

\def\9{{\ \hbox{in}\ \O}}
\def\1{{\ \hbox{on}\ \G_1}}
\def\2{{\ \hbox{on}\ \G_2}}
\def\3{{\ \hbox{on}\ \G_3}}

\def\pa{{\partial}}
\def\pp{{\parallel}}

\begin{document}

\title[An hyperbolic system of P.D.E. relevant in general relativity]{An hyperbolic system of P.D.E. relevant in general relativity}
\author{Giovanni Cimatti}
\address{Department of Mathematics, Largo Bruno
  Pontecorvo 5, 56127 Pisa Italy}
\email{cimatti@dm.unipi.it}






\begin{abstract}
Assuming as starting point the validity of the Einstein-Rosen metric we study the hyperbolic system of P.D.E. to which the Einstein's field equations can be reduced. We prove using the implicit function theorem in Banach spaces, the existence and uniqueness of gravitational waves of small amplitude. A class of solutions, not necessarily small, is also constructed. In the last Section a theorem of existence and uniqueness is given for the corresponding stationary problem.
\end{abstract}

\maketitle

\section{Introduction}
The Einstein's equations of general relativity have attracted the interest of mathematicians since the very beginning of the theory \cite{S} and this interest continues today. Crucial in the development of the theory are the seminal works of Yvonne Choquet-Bruhat \cite{FB1},\cite{CB1}\cite{CB2} and \cite{CB3}. She was able to formulate the problem of the determination of the ten relevant potentials as an initial value problem for an hyperbolic system and to prove the local existence and uniqueness of the solution. The work of Choquet-Bruhat deals with the Einstein's equations in full generality, i.e. without any ``a priori'' restriction on the form of the basic metric. To simplify the study of the field's equations two special geometries are considered: the spherical symmetric case which was the first to be examined \cite{SS}, \cite{SS1}  and remains the most important (see in this respect the results of M. Dafermos \cite{D1}, \cite{D2}, \cite{D3}) and the axis symmetric case which was considered by  Einstein and Rosen \cite{ER}, \cite{R}, \cite{B1}, \cite{B2}, \cite{B3}, \cite{PR}, \cite{WW} and \cite{JW} (pag. 99). In this second case, one assumes as starting point the simplified metric in cylindrical coordinates

\begin{equation}
\label{1_1}
ds^2=e^{2\n-2\m}dt^2-e^{2\n-2\m}dr^2-e^{2\m}dz^2-(r^2e^{-2\m}+e^{2\m}\o^2)d\f^2-2e^{2\m}\o dz d\f,
\end{equation}
where $\m$, $\n$ and $\o$ are real functions of the variables $r$ and $t$. Einstein and Rosen were able to find a class of exact wave solutions of the corresponding field's equations. This was in the past one of the most compelling evidence that general relativity predicts the existence of gravitational waves, a fact today confirmed by experimental observations. If we take $\o=0$ in (\ref{1_1}) we obtain the metric

\begin{equation}
\label{1}
ds^2=e^{2\n-2\m}dt^2-e^{2\n-2\m}dr^2-e^{2\m}dz^2-(r^2e^{-2\m})d\f^2.
\end{equation}
The corresponding field's equations are

\begin{equation}
\label{1_2}
\m_{rr}+\frac{1}{r}\m_r-\m_{tt}=0
\end{equation}

\begin{equation}
\label{2_2}
\n_{r}=r(\m^2_r+\m_t^2)
\end{equation}

\begin{equation}
\label{3_2}
\n_t=2r\m_r\m_t.
\end{equation}
This case for its simplicity has been, and still is, the object, in various contexts, of many papers we quote, among others  \cite{SW}, \cite{SSC}, \cite{SSW}, \cite{CCC} and \cite{TKB}. In this paper we consider the case $\o\neq 0$. This corresponds to two states of polarisation's of cylindrical waves \cite{K}. The system determining $\o(r,t)$ and $\m(r,t)$ is now

\begin{equation}
\label{1_-10}
\o_{tt}-\o_{rr}+\frac{\o_r}{r}=4(\m_r\o_r-\o_t\m_t),\quad r>0
\end{equation}

\begin{equation}
\label{2_-10}
\m_{tt}-\m_{rr}-\frac{\m_r}{r}=\frac{e^{4\m}}{2r^2}(\o_t^2-\o^2_r),
\end{equation}
$\nu(r,t)$ is then determined by a simple integration of an exact differential form.

 By suitably defining $\p$ and $\f$, the system (\ref{1_-10}), (\ref{2_-10}) is shown in Section 5 to be equivalent  to the system

\begin{equation}
\label{3_-10}
\p_{tt}-\p_{rr}-\frac{\p_r}{r}=e^{2\p}(\f_t^2-\f^2_r)
\end{equation}

\begin{equation}
\label{4_-10}
\bigl(e^{2\p}\f_t\bigl)_t-\frac{1}{r}\bigl(e^{2\p}\f_r\bigl)_r=0
\end{equation}
which does not seem to have been studied elsewhere. In all the contributions dealing with the field's equations corresponding to the metric (\ref{1_1}) the problem of possible and mathematically sound boundary conditions is not taken into account \footnote{See however the paper \cite{Reula}}. This question was well present to the mind of Einstein who wrote ``in the first place the boundary conditions presuppose a definite choice of a reference which is contrary to the spirit of relativity''\cite{E}.  However, even taking into account this remark, we think interesting to study the mathematical details of the most typical initial-boundary associated with (\ref{1_-10}) and (\ref{2_-10}). We note also that matter and energy are not distributed uniformly in the universe and this asks for some sort of boundary conditions. 

 To make the paper self-contained we give in Section 2 a detailed derivation of the field's equations corresponding to the metric (\ref{1_1}). 

In Section 3 a result of uniqueness is given for the space flat solution. Whereas in  Section 4 we deal with the existence and uniqueness of exact gravitational waves of small amplitude in the framework of the initial-boundary value problem stated in Section 3. The result is obtained using the inverse function theorem in Banach spaces. The solution obtained in this way exists ``a priori'' only for small initial-boundary data. Thus a natural question arises: do all solutions of (\ref{1_-10}), (\ref{2_-10}) exist for any $t>0$ or certain solutions develop singularities after a finite interval of time ? In Section 5 we present  an example of a class of exact solution of (\ref{1_-10}), (\ref{2_-10}) which are globally defined for $r>0$ and $t>0$, and is different from the class of solutions found by Einstein and Rosen in \cite{ER}. It is an open question if all solutions are equally globally defined. Finally in Section 6 we consider the case in which the unknown functions $\p$ and $\f$ depend only on $r$ and prove a result of existence and uniqueness for the corresponding two point problem.

\section{The Einstein's equations in the stationary axis symmetric case. Derivation of the equations}

We prove here that the field equations corresponding to the metric (\ref{1_1}) are precisely (\ref{1_-10}), (\ref{2_-10}). We note first of all that the only non-vanishing components of the Einstein's symmetric tensor $G_{ij}$ in the present axis  symmetric case are $G_{11},\ G_{22},\ G_{33},\ G_{44},\ G_{12}$ and $ G_{34}$. Thus the system $G_{ij}=0$ is over-determined (as often in general relativity) since we have 6 equations and 3 unknown functions i.e. $\m$, $\n$ and $\o$. We wish to prove that this over-determination is only apparent and that the 6 equations reduce to the two equations (\ref{1_-10}), (\ref{2_-10}) for the determination of $\m(r,t)$ and $\o(r,t)$ and to an exact differential form determining $\n(r,t)$. With this goal in mind we set

\begin{equation*}
A=\frac{1}{r}(r\m_r)_r-\m_{rr},\quad B=\o_{rr}-\frac{1}{r}\o_r-\o_{tt},\quad C=\o_t^2-\o_r^2,
\end{equation*}

\begin{equation*}
D=\m_r\o_r-\m_t\m_t,\quad E=\m_t^2-\m_r^2,\quad F=\n_{tt}-\n_{rr},\quad H=\m_t^2+\m_r^2,\quad L=\o^2_t+\o_t^2,
\end{equation*}

\begin{equation*}
M=8A+\frac{3e^{4\m}}{r^2}C+4(E+F).
\end{equation*}
From $G_{33}=0$, $G_{34}=0$, $G_{44}=0$ we have respectively the equations

\begin{equation}
\label{1_6}
M=0
\end{equation}

\begin{equation}
\label{2_6}
\o r^2M+2r^2(B+4D)=0
\end{equation}

\begin{equation}
\label{3_6}
\o^2r^2(M+4r^2\o^2(B+D))+4r^4e^{-4\m}(E+F)-r^2C=0.
\end{equation}
From (\ref{1_6}) and (\ref{2_6}) we obtain

\begin{equation}
\label{4_6}
B+4D=0
\end{equation}
i.e.

\begin{equation}
\label{5_6}
\o_{tt}-\o_{rr}+\frac{\o_r}{r}=4(\m_r\o_r-\o_t\m_t).
\end{equation}
From (\ref{1_6}) and (\ref{4_6}) we get, by (\ref{3_6}),

\begin{equation}
\label{1_7}
F=\frac{1}{4r^2}e^{4\m}C-F.
\end{equation}
Moreover, from (\ref{1_7}) and (\ref{1_6}) we obtain

\begin{equation}
\label{2_7}
A+\frac{e^4{4\m}}{2r^2}=0
\end{equation}
i.e.

\begin{equation}
\label{3_7}
\m_{tt}-\m_{rr}-\frac{\m_r}{r}=\frac{e^{4\m}}{2r^2}(\o_t^2-\o^2_r).
\end{equation}
The equations (\ref{5_6}) and (\ref{3_7}) form precisely the non-linear system we want to study. It remains to consider the equations $G_{22}=0$, $G_{11}=0$ and $G_{12}=0$. They shall determine $\n(r,t)$. From $G_{12}=0$ we have

\begin{equation}
\label{1_8}
\n_t=2r\m_t\m_r+\frac{1}{2r}e^{4\m}\o_r\o_t
\end{equation}
and from $G_{22}$ or $G_{11}$ we infer

\begin{equation}
\label{2_8}
\n_r=r(\m_t^2+\m_r^2)+\frac{e^{4\m}}{4r}(\o^2_t+\o_r^2).
\end{equation}
We claim that (\ref{1_8}) and (\ref{2_8}) are compatible and determine $\n(r,t)$, apart an arbitrary constant.

\begin{lemma}
If $(\m(r,t),\o(r,t))$ is a solution of the system

\begin{equation}
\label{1_9}
\o_{tt}-\o_{rr}+\frac{\o_r}{r}=4(\m_r\o_r-\o_t\m_t)
\end{equation}

\begin{equation}
\label{2_9}
\m_{tt}-\m_{rr}-\frac{\m_r}{r}=\frac{e^{4\m}}{2r^2}(\o_t^2-\o^2_r)
\end{equation}
the differential form

\begin{equation}
\label{3_9}
F(r,t)dr+G(r,t)dt,
\end{equation}
where

\begin{equation}
\label{4_9}
G(r,t)=2r\m_t\m_r+\frac{1}{2r}e^{4\m}\o_r\o_t
\end{equation}

\begin{equation}
\label{5_9}
F(r,t)=r(\m_t^2+\m_r^2)+\frac{e^{4\m}}{4r}(\o^2_t+\o_r^2)
\end{equation}
is exact.
\end{lemma}

\begin{proof}
We need to prove that $F_t(r,t)-G_r(r,t)=0$ if (\ref{1_9}) and (\ref{2_9}) hold. To this end we add and subtract in $F_t(r,t)-G_r(r,t)$ the quantity

\begin{equation*}
4r^3\m_t\Bigl(\frac{e^{4\m}}{2r^2}\o_t^2+\frac{e^{4\m}}{2r^2}\Bigl).
\end{equation*}
We have, taking into account (\ref{2_9}),

\begin{equation}
\label{1_10}
F_t(r,t)-G_r(r,t)=-2\m_te^{4\m}\o_t^2+e^{4\m}\o_t r\Bigl(\o_{rr}-\o_{tt}-\frac{\o_r}{r}+4\m_r\o_r\Bigl).
\end{equation}
Finally adding and subtracting  $4\m_t\o_t$ from the right hand side of (\ref{1_10}) we obtain, by (\ref{1_9}), $F_t(r,t)-G_r(r,t)=0$ as required.
\end{proof}

\section{An initial-boundary value problem for the system (\ref{1_-10}), (\ref{2_-10}). Uniqueness of the flat-space solution}
A question naturally arises: what are the side conditions which must be added to the system

\begin{equation}
\label{1_12}
\o_{tt}-\o_{rr}+\frac{\o_r}{r}=4(\m_r\o_r-\o_t\m_t)
\end{equation}

\begin{equation}
\label{2_12}
\m_{tt}-\m_{rr}-\frac{\m_r}{r}=\frac{e^{4\m}}{2r^2}(\o_t^2-\o^2_r)
\end{equation}
to obtain a well-posed problem capable of selecting a unique solution of (\ref{1_12}), (\ref{2_12})? We suppose that in the cylinders  $0<r<R_1$ and $R_2<r<\infty$ of the euclidean space referred to cylindrical coordinates is contained all the matter and energy which deform the flat-space metric and consider the problem in the cylinder $R_1<r<R_2$ which is assumed to be empty of matter and energy. The distribution of matter is supposed to be independent of the angular coordinate. We thus have the following initial-boundary conditions for (\ref{1_12}) and (\ref{2_12}):

\begin{equation}
\label{1_14}
\m(R_1,t)=m_1(t),\quad \m(R_2,t)=m_2(t),\quad \m(r,0)=m_0(r),\quad \m_t(r,0)=\tilde m_0(r)
\end{equation}

\begin{equation}
\label{2_14}
\o(R_1,t)=o_1(t),\quad \o(R_2,t)=o_2(t),\quad \o(r,0)=o_0(r),\quad \o_t(r,0)=\tilde o_0(r).\footnote{For example if $R_2/R_1$ is very large and for $r>R_2$  the flat-space solution holds  we would have $m_2(t)=0$, $o_2(t)=0$,  $\tilde m_0(r)=0$ and $\tilde o_0(r)=0$.}
\end{equation}
We start by considering the case in which all the initial-boundary conditions are those corresponding to the flat-space solution. We have the following

\begin{theorem}
The problem

\begin{equation}
\label{1_16}
\o_{tt}-\o_{rr}+\frac{\o_r}{r}=4(\m_r\o_r-\o_t\m_t),\quad r>0
\end{equation}

\begin{equation}
\label{2_16}
\m_{tt}-\m_{rr}-\frac{\m_r}{r}=\frac{e^{4\m}}{2r^2}(\o_t^2-\o^2_r)
\end{equation}

\begin{equation}
\label{3_16}
\m(R_1,t)=0,\quad \m(R_2,t)=0,\quad \m(r,0)=0,\quad \m_t(r,0)=0
\end{equation}

\begin{equation}
\label{4_16}
\o(R_1,t)=0,\quad \o(R_2,t)=0,\quad \o(r,0)=0,\quad \o_t(r,0)=0
\end{equation}
has only the solution $\m(r,t)=0$, $\o(r,t)=0$.
\end{theorem}

\begin{proof}
Let $(\m(r,t),\o(r,t))$ be any solution of (\ref{1_16})-(\ref{4_16}). If we can prove that $\o(r,t)=0$ the result follows. For, in this case the equation (\ref{2_16}) becomes simply  the linear wave equation which with the initial- boundary conditions (\ref{3_16}) implies $\m(r,t)=0$. Let us multiply (\ref{2_16}) by $\o_t$. We have, taking into account the boundary conditions and integrating by parts,

\begin{equation}
\label{2_17}
\frac{1}{2}\int_{R_1}^{R_2}(\o_t^2+\o_r^2)dr=4\int_{R_1}^{R_2}\m_r\o_r\o_t dr-4\int_{R_1}^{R_2}\m_t\o_r^2 dr-\frac{1}{2}\int_{R_1}^{R_2}\o_r\o_t dr.
\end{equation}
Using the Cauchy-Schwartz inequality the right hand side of (\ref{2_17}) can easily be estimated and we obtain

\begin{equation}
\label{3_17}
\frac{1}{2}\int_{R_1}^{R_2}(\o_t^2+\o_r^2)dr\leq\Bigl(2 \sup_{Q_T}|\m_r|+4\sup_{Q_T}|\m_t|+\frac{1}{2R_2}\Bigl)\int_{R_1}^{R_2}(\o_t^2+\o_r^2)dr.
\end{equation}
Define 

\begin{equation*}
\eta(t)=\int_{R_1}^{R_2}(\o_t^2+\o_r^2)dr,\quad C_t=\Bigl(2 \sup_{Q_T}|\m_r|+4\sup_{Q_T}|\m_t|+\frac{1}{2R_2}\Bigl).
\end{equation*}
We have, by (\ref{2_17}) and taking into account the initial-boundary conditions,

\begin{equation*}
\eta'(t)\leq\eta(t),\quad \eta(0)=0.
\end{equation*}
By the Gronwall's theorem \cite{T} we conclude that

\begin{equation*}
\int_{R_1}^{R_2}(\o_t^2(r,t)+\o_r^2(r,t))dr=0\quad \hbox{for all}\ t\in[0,T]
\end{equation*}
and also for all $t\in[0,\infty]$ since $T$ is arbitrary. It follows $\o_t(r,t)=0$ and $\o_r(r,t)=0$ and also $\o(r,t)=0$ in view of (\ref{4_16}) as required.
\end{proof}

\begin{remark}
Recalling that

\begin{equation*}
\n_t=2r\m_t\m_r+\frac{1}{2r}e^{4\m}\o_r\o_t\quad \n_r=r(\m_t^2+\m_r^2)+\frac{e^{4\m}}{4r}(\o^2_t+\o_r^2)
\end{equation*}
we have, under the assumptions of Theorem 2.1, $\n_t=0$, $\n_r=0$. Hence $\n(r,t)$ is constant. Therefore, if, in addition to (\ref{3_16}), (\ref{4_16},) we assume the initial condition $\n(0,0)=0$, we have $\n(r,t)=0$ which together with $\m(r,t)=0$, $\o(r,t)=0$ corresponds to the flat-space solution.
\end{remark}

\section{existence and uniqueness of exact gravitational waves of small amplitude}
Theorem 3.1 suggests to investigate if a branch of non trivial solutions starts from the flat space solution when the initial-boundary data are ``small'' in suitably taken functional spaces. To this end we shall use the inverse function theorem in Banach space which we quote below for the sake of completeness \cite{AP}.

\begin{theorem}
Let $\xx$ and $\yy$ be Banach spaces and $\ff$ a map from $\xx$ to $\yy$ of class $C^1$ such that $\ff(\000)=\000$. Let $\ff'(\000)$ be the Frechet differential of $\ff$. If $\ff'(\000)
$, as a linear map from $\xx$ to $\yy$, is invertible with continuous inverse then there exists a neighbourhood $\nn$ of $\000\in\xx$ and a neighbourhood $\vv$ of $\000\in\yy$ such that $\ff:\nn\to\vv$ is invertible with inverse differentiable.
\end{theorem}

\begin{theorem}
Assume

\begin{equation}
\label{1_25}
m_1(t)\in C^2([0,T]),\ m_2(t)\in C^2([0,T]),\ m_0(t)\in C^2([R_1,R_2]),\ \tilde m_0\in C^1([R_1,R_2])
\end{equation} 

\begin{equation}
\label{2_25}
o_1(t)\in C^2([0,T]),\ o_2(t)\in C^2([0,T]),\ o_0(t)\in C^2([R_1,R_2]),\ \tilde o_0\in C^1([R_1,R_2])
\end{equation} 
with 

\begin{equation}
\label{3_25}
m_1(0)=m_0(R_1),\ m_2(0)=m_0(R_2),\ m'_1(0)=\tilde m_0(R_1),\ m_2'(0)=\tilde m_2(R_2)
\end{equation} 

\begin{equation}
\label{3_26}
o_1(0)=o_0(R_1),\ o_2(0)=o_0(R_2),\ o'_1(0)=\tilde o_0(R_1),\ o_2'(0)=\tilde o_2(R_2). \footnote{These are compatibility conditions needed to obtain a regular solution to the problem.}
\end{equation} 
There exists $\e>0$ such that, if 

\begin{equation}
\label{1_26}
\pp m_1\pp_{C^2([0,T])}+\pp m_2\pp_{C^2([0,T])}+\pp m_0\pp_{C^2([R_1,R_2])}+\pp\tilde m_0\pp_{C^1([R_1,R_2])}+
\end{equation} 

\begin{equation*}
\pp o_1\pp_{C^2([0,T])}+\pp o_2\pp_{C^2([0,T])}+\pp o_0\pp_{C^2([R_1,R_2])}+\pp\tilde o_0\pp_{C^1([R_1,R_2])}< \e,
\end{equation*}
the initial-boundary value problem

\begin{equation}
\label{1_27}
\o_{tt}-\o_{rr}+\frac{\o_r}{r}=4(\m_r\o_r-\o_t\m_t)
\end{equation}

\begin{equation}
\label{2_27}
\m_{tt}-\m_{rr}-\frac{\m_r}{r}=\frac{e^{4\m}}{2r^2}(\o_t^2-\o^2_r)
\end{equation}

\begin{equation}
\label{3_27}
\o(R_1,t)=o_1(t),\ \o(R_2,t)=o_2(t),\ \o(r,0)=o_0(r),\ \o_t(r,0)=\tilde o_0(r)
\end{equation}

\begin{equation}
\label{4_27}
\m(R_1,t)=m_1(t),\ \m(R_2,t)=m_2(t),\ \m(r,0)=m_0(r),\ \m_t(r,0)=\tilde m_0(r)
\end{equation}
has one and only one solution of class $C^2(\bar Q_T)$.
\end{theorem}

\begin{proof}
We apply Theorem 4.1 with the following functional spaces

\begin{equation*}
\xx=C^2(\bar Q_T)\times C^2(\bar Q_T),\quad \yy=\aa\times \bb\times \bb
\end{equation*}
where  $\aa=C^0(\bar Q_T)\times C^0(\bar Q_T)$ and
\vskip .5cm
 $\bb=\bigl\{(A_1(t),A_2(t),A_0(t),B(r));\ A_1(t)\in C^2([0,T]),\ A_2(t)\in C^2([0,T]),\ A_0(t)\in C^2([R_1,R_2]), B(r)\in C^1([R_1,R_2]),\ A_1(0)=A_0(R_1),\ A_2(0)=A_0(R_2),\ 
A_1'(0)=B(R_1),\ A_2'(0)=B(R_2)\bigl\}$.
\vskip .5 cm
 $\bb$ becomes a Banach space with the norm

\begin{equation*}
  \pp \BBB  \pp_\bb=\pp A_1\pp_{C^2([0,T])}+\pp A_2\pp_{C^2([0,T])}+\pp A_0\pp_{C^2([R_1,R_2])}+\pp B\pp_{C^1([R_1,R_2])},\ \BBB \in\bb.
\end{equation*}
$\xx$ is normed with

\begin{equation*}
\pp\XXX  \pp_\xx=\pp X_1\pp_{C^2(\bar Q_T)}+\pp X_2\pp_{C^2(\bar Q_T)},\ \XXX\in\xx.
\end{equation*}
Finally $\yy$ is a Banach spaces with norm $\pp\YYY \pp_\yy=\pp\AAA \pp_\aa+\pp\ \BBB \pp_\bb+\pp\ \CCC \pp_\bb$. Define $\G_1=\{(r,t) ;\ r=R_1,\ 0<t<T\}$, $\G_0=\{(r,t) ;\ t=0,\ R_1<r<R_2\}$,\ $\G_2=\{(r,t) ;\ r=R_2,\ 0<t<T\}$. Let $(\m(r,t),\o(r,t))\in\xx$. The function $\ff:\xx\to\yy$ of Theorem 4.1 has here the form

\begin{equation*}
\ff((\m,\o))=\Bigl( \overbrace {\m_{tt}-\m_{rr}-\frac{\m_r}{r}-\frac{e^{4\m}}{2r^2}(\o_t^2-\o^2_r),\  \o_{tt}-\o_{rr}+\frac{\o_r}{r}-4(\m_r\o_r-\o_t\m_t)}^{in\ \aa},
\end{equation*}

\begin{equation*}
 \underbrace{\m_{|\G_1},\ \m_{|\G_2},\ \m_{|\G_0},\ (\frac{\pa\m}{\pa t})_{|\G_0}}_{in\ \bb},\quad\underbrace{\o_{|\G_1},\quad \o_{|\G_2},\ \o_{|\G_0},\quad (\frac{\pa\o}{\pa t})_{|\G_0}}_{in\ \bb}\Bigl).
\end{equation*}
It is easy to verify that in these spaces the function $\ff$ is well-defined and $\ff(\000)=\000$. If we compute the Frechet's differential of $\ff$ in $\000\in\xx$ we find, with $(U,V)\in \xx$

\begin{equation*}
\ff'(\000)[U,V]=\Bigl(\overbrace{(U_{tt}-U_{rr}-\frac{1}{r}U_r,\quad V_{tt}-V_{rr}-\frac{1}{r}V_r)}^{in\ \aa},
\end{equation*}

\begin{equation*}
 \bigl(\underbrace{U_{|\G_1}, \ U_{|\G_2},\ U_{|\G_0},\ \bigl(\frac{\pa U}{\pa t}\bigl)_{|\G_0})}_{in\ \bb},\ (\underbrace{V_{|\G_1}, \ V_{|\G_2},\ V_{|\G_0},\ \bigl(\frac{\pa V}{\pa t}\bigl)_{|\G_0}}_{in\ \bb})\Bigl).
\end{equation*}
To apply Theorem 4.1 we need to prove that $\ff'$ is invertible as a linear operator from $\xx$ to $\yy$. To this end we consider the linear problem 

\begin{equation}
\label{1_31}
\ff'({\bf 0})[U,V]=\YYY\in\yy,
\end{equation}
where $\YYY=\Bigl(\overbrace{\bigl(f(r,t),g(r,t)\bigl)}^{in\ \aa},\ \bigl(\underbrace{n_1(t),n_2(t),n_0(r),n(r)\bigl)}_{in\ \bb},\  \bigl( \underbrace{l_1(t),l_2(t),l_0(r),l(r)\bigl)\Bigl)}_{in\ \bb})$. In components (\ref{1_31}) reads

\begin{equation}
\label{1_32}
U_{tt}-U_{rr}-\frac{1}{r}U_r=f(r,t)\quad \hbox{in}\quad Q_T
\end{equation}

\begin{equation}
\label{2_32}
U_{|G_1}=n_1(t),\quad U_{|G_2}=n_2(t),\quad U_{|G_0}=n_0(r),\quad \bigl(\frac{\pa U}{\pa t}\bigl)_{|G_0}=n(r)
\end{equation}

\begin{equation}
\label{3_32}
V_{tt}-V_{rr}-\frac{1}{r}V_r=g(r,t)\quad \hbox{in}\quad Q_T
\end{equation}

\begin{equation}
\label{4_32}
V_{|G_1}=l_1(t),\quad V_{|G_2}=l_2(t),\quad V_{|G_0}=l_0(r),\quad \bigl(\frac{\pa U}{\pa t}\bigl)_{|G_0}=l(r).
\end{equation}
The problems (\ref{1_32}), (\ref{2_32}) and  (\ref{3_32}), (\ref{4_32}) are uncoupled.  (\ref{1_32}), (\ref{2_32}) is simply the initial-boundary value for the wave equation in cylindrical coordinates with non-vanishing right hand side and non homogeneous initial boundary conditions. This problem is certainly solvable with continuous dependence from the data. The same can be said of  (\ref{3_32}), (\ref{4_32}). Thus problem (\ref{1_31}) is solvable with a unique solution and with continuous dependence on the data. Therefore Theorem 4.1 is applicable
\end{proof}

\section{A class of exact solutions of the system (\ref{1_27}), (\ref{2_27})}
 In this Section we construct a class of exact solutions which are globally defined in time. We first rewrite (\ref{1_27})-(\ref{2_27}) in divergence form. We proceeds in three steps. First of all it is easy to verify, by direct computation, that

\begin{equation}
\label{1_35}
\Bigl(\frac{e^{4\m}}{r}\o_t\Bigl)_t-\Bigl(\frac{e^{4\m}}{r}\o_r\Bigl)_r=0
\end{equation}

\begin{equation}
\label{2_35}
\m_{tt}-\m_{rr}-\frac{1}{r}\m_r=\frac{e^{4\m}}{2r^2}\bigl(\o_t^2-\o_r^2\bigl)
\end{equation}
is fully equivalent to the system (\ref{1_27}), (\ref{2_27}). Secondly it is convenient to use the simple transformation $\m=\frac{\p}{2}+\frac{1}{2}\log r$, $\o=\p$ which for $r>0$ is a diffeomorphism. In terms of $\p$ and $\phi$, (\ref{1_35}), (\ref{2_35}) becomes 

\begin{equation}
\label{1_36}
\p_{tt}-\frac{1}{r}\big(r\p_r\bigl)_r=e^{2\p}(\f_t^2-\f^2_r)=0
\end{equation}

\begin{equation}
\label{2_36}
\bigl(e^{2\p}\f_t\bigl)_t-\frac{1}{r}\bigl(re^{2\p}\f_r\bigl)_r=0.
\end{equation}
Thirdly we give a divergence form also to the equation (\ref{1_36}). To this end we note that, by (\ref{2_36}) we have

\begin{equation}
\label{1_37}
re^{\p}\bigl(\f_r^2-\f_t^2\bigl)=\bigl[\f\bigl(re^{2\p}\f_r\bigl]_r-\bigl[\f\bigl(re^{2\p}\f_t\bigl]_t.
\end{equation}
Thus the system (\ref{1_36}), (\ref{2_36}) takes the full divergence form

\begin{equation}
\label{2_37}
\bigl[r\bigl(\p_r-\f e^{2\p}\f_r\bigl)\bigl]_r-\bigl[r\bigl(\p_t-\f e^{2\p}\f_t\bigl)\bigl]_t=0
\end{equation}

\begin{equation}
\label{3_37}
\bigl(e^{2\p}\f_t\bigl)_t-\frac{1}{r}\bigl(e^{2\p}\f_r\bigl)_r=0.
\end{equation}
We restrict the search of solutions of the system (\ref{2_37}), (\ref{3_37}) to solutions for which there exist a functional relation of the form:

\begin{equation}
\label{4_37}
\p=\Psi(\f).
\end{equation}
 See for this approach \cite{C1}, \cite{C2}. The equations (\ref{2_37}) and (\ref{3_37}) become

\begin{equation}
\label{1_38}
\Bigl(re^{2\Psi(\f)}\f_r\Bigl)_r-\Bigl(re^{2\Psi(\f)}\f_t\Bigl)_t=0
\end{equation}

\begin{equation}
\label{2_38}
\Bigl[r\bigl(\Psi'(\f)-\f e^{2\Psi(\f)}\bigl)\f_r\Bigl]_r-\Bigl[r\bigl(\Psi'(\f)-\f e^{2\Psi(\f)}\bigl)\f_t\Bigl]_t=0.
\end{equation}

Let $\Psi(\f)$ be the solution of the ordinary differential equation

\begin{equation}
\label{3_38}
\frac{d\Psi}{d\f}-\f e^{2\Psi}=\g e^{2\Psi},\quad \g\ \hbox{ a real parameter}.
\end{equation}
Equation (\ref{3_38}) is easily solved. We find

\begin{equation}
\label{1_39}
e^{-2\Psi}=-\bigl(\f^2+2\g\f+C\bigl),
\end{equation}
where $C$ is the constant of integration. We have from (\ref{1_39}) the compatibility condition

\begin{equation}
\label{2_39}
-\bigl(\f^2+2\g\f+C\bigl)>0
\end{equation}
which is verified if

\begin{equation}
\label{4_39}
-\g-\sqrt{\g^2-C}<\f<-\g+\sqrt{\g^2-C}\quad \hbox{and}\quad C<\g^2.
\end{equation}
If  (\ref{4_39}) is satisfied we have

\begin{equation}
\label{5_39}
e^{2\Psi}=-\frac{1}{\f^2+2\g\f+C}.
\end{equation}
Substituting (\ref{5_39}) in (\ref{1_38}) we obtain

\begin{equation}
\label{1_40}
\frac{1}{r}\Bigl[r\frac{\f_r}{\f^2+2\g\f+C}\Bigl]_r-\Bigl[\frac{\f_t}{\f^2+2\g\f+C}\Bigl]_t=0.
\end{equation}
We wish to linearised (\ref{1_40}). Let us define

\begin{equation}
\label{1_10_1}
\Theta=F(\Phi):=\frac{-1}{2\sqrt{\g^2-C}}\log{\Bigl[\frac{-\Phi-\g+\sqrt{\g^2-C}}{\Phi+\g+\sqrt{\g^2-C}}\Bigl]}.
\end{equation}
We have, if (\ref{4_39}) holds

\begin{equation}
\label{2_10_1}
\frac{dF(\Phi)}{d\Phi}=-\frac{1}{\Phi^2+2\g+C}>0.
\end{equation}
 Let $\theta(r,t)=F(\phi(r,t))$. We have

\begin{equation}
\label{5_10_1}
\theta_r(r,t)=-\frac{\phi_r}{\phi^2+2\g\phi+C},\quad \theta_t(r,t)=-\frac{\phi_t}{\phi^2+2\g\phi+C}.
\end{equation}
Hence (\ref{1_40}) becomes the linear wave equation in cylindrical coordinates i.e.

\begin{equation}
\label{1_10_2}
\frac{1}{r}\bigl(r\theta_r\bigl)_r-\theta_{tt}=0.
\end{equation}
Let $\theta(r,t)$ be any solution of equation (\ref{1_10_2}). \footnote{We could take for example $\theta=K_0(r)e^t$ where $K_0(r)$ denotes the modified Bessel function of the second type, or $\theta(r,t)=J_0(r,t)sin(t)$, where $J_0(r,t)$ is the Bessel function of the first type.} Since $F(\Phi)$ maps one-to-one the open interval $(-\g-\sqrt{\g^2-C},-\g+\sqrt{\g^2-C})$ onto $(-\infty-\infty)$ the function $\Theta=F(\Phi)$ is globally invertible. Precisely we have

\begin{equation}
\label{1_10_3}
\Phi=F^{-1}(\Theta)=\frac{(\sqrt{\g^2-C}-\g)e^{2\Theta\sqrt{\g^2-C}}-(\sqrt{\g^2-C}+\g)}{e^{2\Theta\sqrt{\g^2-\g}}+1}.
\end{equation}
In conclusion we find the following class of solutions of the system (\ref{1_36}), (\ref{2_36})

\begin{equation}
\label{1_10_4}
\phi(r,t)=\frac{(\sqrt{\g^2-C}-\g)e^{2\sqrt{\g^2-C}\theta(r,t)}-(\sqrt{\g^2-C}+\g)}{e^{2\theta(r,t)\sqrt{\g^2-\g}}+1}
\end{equation}
with the corresponding $\psi(r,t)$ given by

\begin{equation}
\label{2_10_4}
\psi(r,t)=\log{\frac{1}{\sqrt{-(\phi^2(r,t)+2\g\phi(r,t)+C)}}},
\end{equation}
where $\theta(r,t)$ is any solution of (\ref{1_10_2}).

 An open question naturally arises: are all solutions of (\ref{1_36}), (\ref{2_36}) globally defined as it happens for (\ref{1_10_4}), (\ref{2_10_4}) or, in certain cases, they develop singularities in a finite time? This second possibility appears as more reasonable in view of the quadratic non-linearity entering in the right hand side of (\ref{1_36}).

\section{existence and uniqueness for the stationary problem}
In this last Section we study the stationary counterpart of the initial-boundary value problem (\ref{1_27})-(\ref{4_27}) expressed in terms of $(\f,\p)$. Thus we consider the two-point problem

\begin{equation}
\label{1_45}
\frac{1}{r}\bigl(r e^{2\p}\f'\bigl)'=0\quad \hbox{in}\quad (R_1,R_2)
\end{equation}

\begin{equation}
\label{2_45}
\f(R_1)=\f_1,\quad \f(R_2)=\f_2
\end{equation}

\begin{equation}
\label{3_45}
\frac{1}{r}\bigl(r \p'\bigl)'=e^{2\p}\f'^2\quad \hbox{in}\quad (R_1,R_2)
\end{equation}

\begin{equation}
\label{4_45}
\p(R_1)=\p_1,\quad \p(R_2)=\p_2.
\end{equation}
It is not restrictive to assume $\f_1=0$ in (\ref{1_45})-(\ref{4_45}). For  $(\f(r)+C,\p(r))$ is still a solution ($C$ an arbitrary constant) if $(\f(r),\p(r))$ is a solution of (\ref{1_45})-(\ref{4_45}). Therefore, if $(\f(r),\p(r))$ is the solution corresponding to the case $\f_1=0$, the solution for $\f_1\neq 0$ is simply  $(\f(r)+\f_1,\p(r))$. Hence we shall study the problem

\begin{equation}
\label{1_47}
\frac{1}{r}\bigl(r e^{2\p}\f'\bigl)'=0\quad \hbox{in}\quad (R_1,R_2)
\end{equation}

\begin{equation}
\label{2_47}
\f(R_1)=0
\end{equation}

\begin{equation}
\label{3_47}
 \f(R_2)=a
\end{equation}

\begin{equation}
\label{4_47}
\frac{1}{r}\bigl(r \p'\bigl)'=e^{2\p}\f'^2\quad \hbox{in}\quad (R_1,R_2)
\end{equation}

\begin{equation}
\label{5_47}
\p(R_1)=b
\end{equation}

\begin{equation}
\label{6_47}
\p(R_2)=c.
\end{equation}
The case $a=0$ is immediately dealt with. For multiplying (\ref{1_47}) by $\f$, integrating by parts and taking into account (\ref{2_47}) and (\ref{3_47}) we have

\begin{equation}
\label{7_47}
\int_{R_1}^{R_2}e^{2\p}\f'^2 dr=0.
\end{equation}
This in view of (\ref{2_47}) implies $\f(r)=0$. We have from (\ref{4_47})

\begin{equation}
\label{1_48}
\frac{1}{r}\bigl(r \p'\bigl)'=0,\quad \p(R_1)=b,\quad \p(R_2)=c,
\end{equation}
a linear problem which is easily solved. Moreover, we can assume $a>0$. For, if $(\f(r),\p(r))$ is the solution corresponding to $a>0$, the solution for the case $a<0$  is simply $(-\f(r),\p(r))$. By the one-dimensional maximum principle \cite{PW}, we have from (\ref{1_47}), (\ref{2_47})

\begin{equation}
\label{2_49}
0\leq \f(r)\leq a\quad \hbox{in}\quad  [R_1,R_2].
\end{equation}
The transformation

\begin{equation}
\label{0_50}
\ttt=\frac{\f^2}{2}+\frac{1}{2}\bigl(e^{-2\p}-e^{-2b}\bigl)
\end{equation}
 permits to rewrite the problem (\ref{1_47})-(\ref{6_47}) in the more symmetric form

\begin{equation}
\label{1_50}
\frac{1}{r}\bigl(re^{2\p}\ttt'\bigl)=0
\end{equation}

\begin{equation}
\label{2_50}
\frac{1}{r}\bigl(re^{2\p}\f'\bigl)=0
\end{equation}

\begin{equation}
\label{3_50}
\ttt(R_1)=0
\end{equation}

\begin{equation}
\label{4_50}
\ttt_2=:\ttt(R_2)=\frac{1}{2}a^2+\frac{1}{2}e^{-2c}-\frac{1}{2}e^{-2b}
\end{equation}

\begin{equation}
\label{5_50}
\f(R_1)=0
\end{equation}

\begin{equation}
\label{6_50}
\f(R_2)=a,
\end{equation}
where $\ttt$, $\f$, $\p$ are linked by the functional relation (\ref{0_50}). We have

\begin{lemma}
If $(\f(r),\ttt(r),\p(r))$ is a solution of (\ref{0_50})-(\ref{6_50}), then 

\begin{equation}
\label{1_51}
\ttt(r)=\frac{\ttt_2}{a}\f(r).
\end{equation}
\end{lemma}

\begin{proof}
Let $\z(r)=\ttt(r)-\frac{\ttt_2}{a}\f(r)$, we have

\begin{equation}
\label{2_51}
\z(R_1)=0,\quad \z(R_2)=0
\end{equation}
and by (\ref{1_50}) and (\ref{2_50})

\begin{equation}
\label{4_51}
\frac {1}{r}\bigl(re^{2\p}\z'\bigl)'=0.
\end{equation}
Multiplying (\ref{4_51}) by $\z(r)$ and using (\ref{2_51}) we obtain

\begin{equation}
\label{1_52}
\int_{R_1}^{R_2}re^{2\p}\z'^2 dr=0.
\end{equation}
By (\ref{2_51}) this implies $\z(r)=0$ and (\ref{1_51}).
\end{proof}

\noindent From (\ref{0_50}) and (\ref{1_51}) we have

\begin{equation}
\label{2_52}
\frac{\ttt_2}{a}\f(r)=\frac{\f^2(r)}{2}+\frac{1}{2}e^{-2\p(r)}-\frac{1}{2}e^{-2b}
\end{equation}
and, by (\ref{4_50}),

\begin{equation}
\label{3_52}
-\f^2+\big[a+\frac{1}{a}\bigl(e^{-2c}-e^{-2b}\bigl)\bigl]\f+e^{-2b}=e^{-2\p}.
\end{equation}
We need to solve (\ref{3_52}) with respect to $\p$. This requires the positivity of the left hand side  of (\ref{3_52}). We use the following elementary

\begin{lemma}
If $a>0$, $b\in \R$, $c\in\R$ and $0\leq\f\leq a$ we have

\begin{equation}
\label{1_53}
f(\f,a,b,c):=-\f^2+\big[a+\frac{1}{a}\bigl(e^{-2c}-e^{-2b}\bigl)\bigl]\f+e^{-2b}>0.
\end{equation}
\end{lemma}

\begin{proof}
The two roots of the equation $f(\f,a,b,c)=0$ are

\begin{equation}
\label{2_53}
\f_1=\frac{1}{2a}\Bigl[\bigl(e^{-2c}-e^{-2b}+a^2\bigl) + \sqrt{\bigl(e^{-2c}-e^{-2b}+a^2\bigl)^2+4a^2e^{-2b}}\Bigl]
\end{equation}

\begin{equation}
\label{3_53}
\f_2=\frac{1}{2a}\Bigl[\bigl(e^{-2c}-e^{-2b}+a^2\bigl) -\sqrt{\bigl(e^{-2c}-e^{-2b}+a^2\bigl)^2+4a^2e^{-2b}}\Bigl].
\end{equation}
Thus 

\begin{equation}
\label{4_53}
\f_1<0<\f_2.
\end{equation}
Moreover we have

\begin{equation}
\label{1_54}
\f_2-a=\frac{1}{2a}\Bigl[\bigl(e^{-2c}-e^{-2b}-a^2\bigl) +\sqrt{\bigl(e^{-2c}-e^{-2b}+a^2\bigl)^2+4a^2e^{-2b}}\Bigl].
\end{equation}
Hence $\f_2-a>0$ and (\ref{1_53}) follows.
\end{proof}
Let us define

\begin{equation}
\label{0_54}
M(a,b,c)=a+\frac{1}{a}\bigl(e^{-2c}-e^{-2b}\bigl).
\end{equation}
By Lemma 6.2 the equation

\begin{equation}
\label{2_54}
-\f^2+M(a,b,c)\f+e^{-2b}=e^{-2\p}
\end{equation}
in the unknown $\p$ has one and only one solution. Moreover, since

\begin{equation}
\label{3_54}
e^{2\p}=\frac{1}{-\f^2+M\f+e^{-2b}},
\end{equation}
we can restate the equation (\ref{2_50}) in term of the sole $\f$ and we have, for the determination of $\f(r)$, the nonlinear two-point problem

\begin{equation}
\label{1_55}
\Bigl(\frac{r\f'}{-\f^2+M\f+e^{-2b}}\Bigl)'=0
\end{equation}

\begin{equation}
\label{2_55}
\f(R_1)=0,\quad \f(R_2)=a.
\end{equation}
Let us define

\begin{equation}
\label{3_55}
u=\uu(\f):=\int_0^\f\frac{dt}{-t^2+Mt+e^{-2b}}.
\end{equation}
Since $0\leq\f\leq a$, the denominator in the integral (\ref{3_55}) never vanishes and is strictly positive by Lemma 6.2.\footnote{Note that the integral in (\ref{3_55}) can be explicitly computed.}
In terms of $u$ the problem (\ref{1_55}), (\ref{2_55})  becomes simply

\begin{equation}
\label{1_57}
(ru')'=0,\quad u(R_1)=0,\quad u(R_2)=\uu(a),
\end{equation}
a linear problem which is immediately solved. On the other hand $\uu'(\f)>0$ in $[0,a]$. Thus $\uu(\f)$ maps $[0,a]$ one-to-one onto $[0,\uu(a)]$. We conclude that , if $u(r)$ is the solution of (\ref{1_57}), we have

\begin{equation}
\label{2_57}
\f(r)=\uu^{-1}(u(r)).
\end{equation}
Finally, the corresponding $\p(r)$ is, by (\ref{3_54}),

\begin{equation}
\label{3_57}
\p(r)=-\frac{1}{2}\log\bigl[-\f^2(r)+M\f(r)+e^{-2b}\bigl].
\end{equation}

\begin{remark}
The solution of problem (\ref{1_50})-(\ref{6_50}) obtained above is also unique since the solution of the linear problem (\ref{1_57}), to which the starting problem  (\ref{1_50})-(\ref{6_50}) is reduced, is surely unique.
\end{remark}

\bibliographystyle{amsplain}

\end{document}